\documentclass[a4paper, 11pt]{article}

\usepackage{babel}
\usepackage[utf8]{inputenc}
\usepackage[T1]{fontenc}
\usepackage[margin=1.15in]{geometry}
\usepackage{microtype}

\usepackage{amsmath, amsthm, amssymb}
\usepackage{enumerate}
\usepackage{thmtools}
\usepackage{thm-restate}
\usepackage{mathtools}
\usepackage{hyperref}
\usepackage[ruled,vlined,linesnumbered]{algorithm2e}
\usepackage{xspace}
\usepackage{aligned-overset}
\usepackage{textpos}

\usepackage{tikz}
\usetikzlibrary{positioning}

\declaretheorem[name=Theorem]{theorem}

\declaretheorem[name=Lemma, sibling=theorem]{lemma}

\declaretheorem[name=Observation, sibling=theorem]{observation}

\usepackage{hyperref}
\usepackage[nameinlink]{cleveref}
\hypersetup{
  colorlinks=true,
  linkcolor=blue!80!black,
  citecolor=blue!80!black,
  urlcolor=green!35!black
}

\crefname{claim}{Claim}{Claims}
\crefname{case}{case}{cases}
\crefname{algorithm}{Algorithm}{Algorithms}
\crefname{theorem}{Theorem}{Theorems}
\crefname{figure}{Figure}{Figures}
\crefname{lemma}{Lemma}{Lemmas}
\crefname{observation}{Observation}{Observations}

\let\le\leqslant
\let\ge\geqslant
\let\leq\leqslant
\let\geq\geqslant

\newcommand{\BS}{\textsc{RecSpan}}
\newcommand{\Z}{\textsc{Zorro}}
\newcommand{\bigO}{\mathcal{O}}

\usepackage[]{todonotes}
\graphicspath{{./}}

\title{Almost Linear 3-Spanners of Temporal Cliques\thanks{This research has been initiated/conducted at the AlgUW workshop (B\k{e}dlewo 09.2025), supported by the Excellence
Initiative -- Research University (IDUB) funds of the University of Warsaw. \\ 
\textbf{Funding information:}\\
Davide Bil\`o acknowledges the support of the project \emph{SOS-TG: Spanners and Oracles for Static and Temporal Graphs}, funded by Università degli Studi dell'Aquila under the Call for Proposals for Fundamental Research and Early-Career Research Grants -- Year 2026.\\
Júlia Baligács (during employment in Warsaw), V\'aclav Bla\v zej, Ma\"el Dumas, and Anna Zych-Pawlewicz: ERC project BOBR, funded from the European Research Council (ERC) under the European Union's Horizon 2020 research and innovation programme with grant agreement No. 948057.\\
Júlia Baligács (during employment in Oxford): ERC grant CCOO (grant no.~101165139). Views and opinions expressed are however those of the authors only and do not necessarily reflect those of the European Union or the European Research Council. Neither the European Union nor the granting authority can be held responsible for them.\\
}
}

\date{}
\newcommand{\aff}[1]{\textcolor{black!50}{#1}}

\author{
    Júlia Baligács\\{\small\aff{University of Oxford}}\\
    \href{mailto:jbaligacs@gmail.com}{\small jbaligacs@gmail.com}\bigskip
    \and
    Davide Bil\`o\\{\small\aff{University of L'Aquila}}\\
    \href{mailto:davide.bilo@univaq.it}{\small davide.bilo@univaq.it}
    \and
     V\'aclav Bla\v zej\\{\small\aff{University of Warsaw}}\\
    \href{mailto:vaclavblazej@gmail.com}{\small vaclavblazej@gmail.com}
    \and
    Ma\"el Dumas\\{\small\aff{University of Warsaw}}\\
    \href{mailto:mdumas@mimuw.edu.pl}{\small mdumas@mimuw.edu.pl}
    \and
    Anna Zych{-}Pawlewicz\\{\small\aff{University of Warsaw}}\\
    \href{mailto:anka@mimuw.edu.pl}{\small anka@mimuw.edu.pl}
}

\begin{document}
\begingroup
\renewcommand{\footnotemark}{}
\maketitle
\endgroup
\thispagestyle{empty}
\begin{textblock}{20}(-1.9, 7)
   \includegraphics[width=40px]{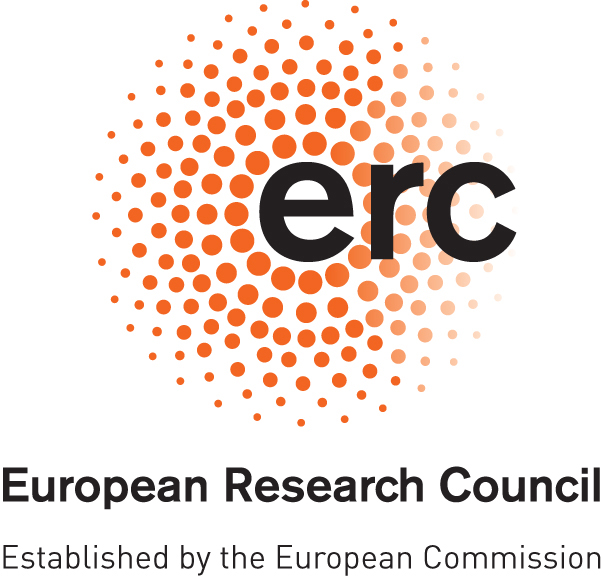}\end{textblock}
\begin{textblock}{20}(-2.15, 7.4)
    \includegraphics[width=60px]{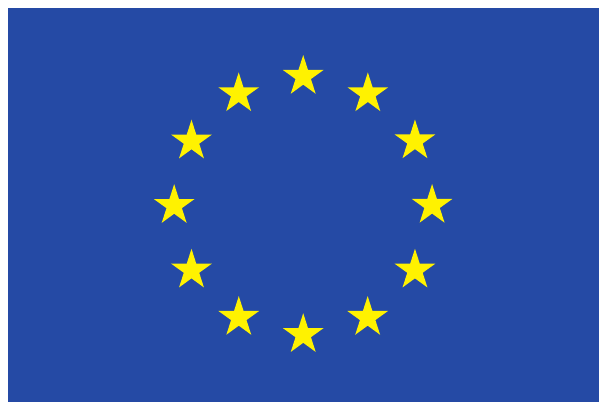}\end{textblock}

\begin{abstract}
Temporal graphs model dynamic networks by assigning positive integer time labels to the edges, while information propagates along \emph{temporal paths}, whose edge labels are traversed in nondecreasing order. A \emph{temporal $\alpha$-spanner} of a temporal graph with $n$ vertices is a temporal subgraph that approximates the minimum-hop temporal distance between every pair of vertices within a factor of $\alpha$. While general temporal graphs may not admit sparse temporal $\alpha$-spanners for any value of $\alpha$, temporal cliques are known to admit temporal $(2k-1)$-spanners of size $\widetilde{\mathcal{O}}(kn^{1+1/k})$ for every positive integer $k$.

We present a simple recursive algorithm that computes, for every temporal clique on $n$ vertices, a temporal $3$-spanner of size $n^{1+2/\sqrt{\ln n}}=n^{1+o(1)}$, thereby improving the previous best upper bound of $\widetilde{\mathcal{O}}(n^{3/2})$. 
We also show that a modified version of our algorithm computes temporal $3$-spanners of size $\mathcal{O}(nL)$ when the lifetime is bounded by~$L$, 
i.e., all time labels are in $\{1,\ldots,L\}$,
thus improving the previous bound of~$\mathcal{O}(2^Ln\log n)$. Both results are particularly striking in light of the known lower bound of~$\Omega(n^2)$ on the size of temporal $2$-spanners, which already holds for temporal cliques of lifetime $L\geq 3$.

Both algorithms rely on a new simple recursive decomposition that certifies temporal connectivity for a large collection of source--target pairs using only $\mathcal{O}(n)$ carefully selected edges and recursively processes only the remaining pairs. Besides yielding substantially improved upper bounds, this approach is significantly simpler than previous constructions.
\end{abstract}

\newpage
\setcounter{page}{1}
\section{Introduction}

\emph{Temporal graphs} provide a natural model for networks whose topology changes over time. In this framework, edges are assigned one or more integer \emph{time labels}, and information propagates along \emph{temporal paths}, i.e., paths whose edge labels are traversed in nondecreasing order.
More formally, a \emph{temporal graph} $G=(V,E,\lambda)$ consists of a set of vertices $V$, a set of edges~$E\subseteq \binom{V}{2}$, and a time labeling function $\lambda \colon E \to \mathcal{P}(\mathbb{N})$.
For two vertices $u,v\in V$, a \emph{temporal $u$-$v$-path} is a $u$-$v$-path in the underlying graph $(V,E)$ such that the sequence of traversed edges~$e_1,\dots, e_k$ admits labels $\lambda_i\in \lambda(e_i)$ ($i \in \{1, \dots, k\}$) with $\lambda_1\leq \ldots \leq \lambda_k$. The \emph{size} of a temporal graph is the total number of labels, that is, $\sum_{e\in E}|\lambda(e)|$.
Temporal graphs arise naturally in many dynamic systems, including communication networks, transportation systems, social interactions, and biological networks~\cite{CFQS12,Mic16}.

A fundamental question in temporal graph algorithms is how to sparsify such networks while preserving temporal connectivity. Given a temporal graph that is temporally connected, i.e., it contains a temporal path between every ordered pair of vertices, a \emph{temporal spanner} is 
a spanning subgraph that preserves temporal connectivity. 
Formally, to specify a temporal spanner of~$G=(V,E,\lambda)$,
it suffices to specify a
set of edges $E'\subseteq E$ and a labeling
function~$\lambda' \colon E' \to \mathcal{P}(\mathbb{N})$
with~$\lambda'(e)\subseteq \lambda(e)$ for every $e\in E'$. The
size
of the corresponding spanner is then the size of $G'=(V,E',\lambda')$.
Unlike static graphs, where sparse connectivity certificates always exist
(namely a spanning tree),
temporal graphs may require dense spanners. Indeed, Kempe, Kleinberg, and Kumar~\cite{KKK02} constructed temporal graphs on~$n$ vertices for which every temporal spanner has size $\Omega(n\log n)$. Axiotis and Fotakis~\cite{AF16} later strengthened this lower bound to $\Omega(n^2)$.
This motivates the study of graph classes where temporal connectivity can still be preserved by sparse certificates.

Temporal cliques are the canonical positive example in this setting, as they admit remarkably sparse temporal spanners despite containing $\Theta(n^2)$ edges, where $n$ is the number of vertices. Casteigts, Peters, and Schoeters~\cite{CPS21} were the first to design an algorithm computing temporal spanners of size $\bigO(n\log n)$ for temporal cliques. Later, Carnevale, Casteigts, and Corsini~\cite{CCC25} gave a significantly simpler recursive construction achieving the same asymptotic bound. More recently, Baligács~\cite{Bal26} showed that every temporal clique admits a temporal spanner of size at most $7n$. Finally, Angrick~\emph{et~al.}~\cite{ABF+24} proved that temporal cliques and temporal bicliques admit temporal spanners of asymptotically the same size. Consequently, sparse temporal spanners are now essentially understood for temporal cliques and bicliques. Beyond worst-case graph classes, Casteigts, Raskin, Renken, and Zamaraev~\cite{CRRZ24} studied random temporal graphs generated from Erd\H{o}s--Rényi random graphs and showed that they admit nearly optimal temporal spanners asymptotically almost surely.

A natural next question is whether similarly sparse certificates can preserve not only temporal connectivity, but also the minimum-hop temporal distances. In many applications, the number of intermediate relays is a critical parameter: every additional hop increases communication overhead and may introduce additional sources of failure. This motivates the study of \emph{temporal $\alpha$-spanners}, which approximate the minimum-hop temporal distance between every pair of vertices within a multiplicative factor of $\alpha$. 
For temporal cliques, where every pair of vertices is directly connected by a temporal edge, a temporal $\alpha$-spanner guarantees 
a temporal path with at most $\alpha$ edges between every pair of vertices.

Bilò, D'Angelo, Gualà, Leucci, and Rossi~\cite{BDGLR22} initiated the study of temporal $\alpha$-spanners, and showed that temporal cliques admit temporal $(2k-1)$-spanners of size $\bigO(kn^{1+1/k}\log^{1-1/k}n)$ for every integer $k\ge 1$. In particular, they obtained temporal $\lfloor\log n\rfloor$-spanners of size $\bigO(n\log^2 n)$. They also studied temporal cliques with \emph{lifetime} $L$, i.e., instances in which all time labels belong to $\{1,\ldots,L\}$, and obtained temporal $3$-spanners of size~$\bigO(2^Ln\log n)$. They further proved that temporal $2$-spanners have size $\bigO(n\log n)$ for $L=2$, but require $\Omega(n^2)$ edges for $L\geq 3$. Therefore, temporal $3$-spanners are the first non-trivial constant-hop regime for $L\geq 3$ admitting sparse spanners.

Despite these results, no superlinear lower bound is known on the size of temporal $\alpha$-spanners for temporal cliques even for $\alpha=3$. Understanding the tradeoff between sparsity and hop bounds is therefore a fundamental open problem, which is the focus of this work.

\paragraph{Our contribution.}

We show that every temporal clique admits a temporal $3$-spanner of almost linear size.

\begin{theorem}
\label{thm:3-spanners}
Every temporal clique on $n$ vertices admits a temporal $3$-spanner of size at most $n^{1+2/\sqrt{\ln n}}=n^{1+o(1)}$. Moreover, such a spanner can be computed in time $\bigO(n^{2+o(1)})$.
\end{theorem}

This is the first construction of an almost-linear-size temporal spanner with a constant hop bound, improving the previous best upper bound of $\widetilde{\bigO}(n\sqrt n)$ for 3-spanners. 
We also consider temporal cliques with bounded lifetime $L$ and show that they admit sparse temporal~$3$-spanners.

\begin{theorem}
\label{thm:lifetime}
Every temporal clique on $n$ vertices of lifetime $L$ admits a temporal 3-spanner of size~$\bigO(nL)$.
Moreover, such a spanner can be computed in time $\bigO(n^3L)$.
\end{theorem}

This improves the previous bound of $\bigO(2^Ln\log n)$ by removing the logarithmic factor and replacing the exponential dependence on the lifetime with a linear one. It is worth noticing that all our results extend to temporal bicliques as well.

While previous constructions relied on substantially more involved recursive de\-com\-po\-si\-tions, our algorithms are based on a simple recursive procedure, called ~$\Z$, that exploits the equivalence between temporal cliques and temporal bicliques established in~\cite{ABF+24}. Our procedure repeatedly identifies a carefully selected set of $\bigO(n)$ edges that simultaneously certifies $3$-hop temporal connectivity for a large collection of source--target pairs. The key idea is that these pairs do not need to be handled individually: after removing all pairs already certified by our procedure, the algorithm recursively processes only the remaining pairs. This leads to a simple recursive decomposition that captures the structure of temporal cliques and achieves near-linear sparsity.

Our results reveal a sharp transition in the hop-constrained setting: while temporal $2$-spanners may require quadratic size, allowing one additional hop already enables near-linear sparsification. An intriguing open question is whether every temporal clique admits a linear-size temporal $3$-spanner.

 \section{Preliminaries}

We begin by collecting the notation, conventions, and results used throughout the paper.

First, observe that it suffices to prove \cref{thm:3-spanners,thm:lifetime} for temporal cliques in which every edge carries exactly one time label: Indeed, given a temporal clique $G$ in which an edge may carry multiple time labels, retain an arbitrary single label on each edge. Any temporal $\alpha$-spanner of the resulting temporal clique is also a temporal $\alpha$-spanner of the original.
Throughout the remainder of the paper, we may therefore assume that
every edge carries exactly one time label, i.e., $\lambda \colon E(G) \to \mathbb{N}$
is a function assigning a natural number to every edge in $G$.
Moreover, for our purposes, a temporal spanner is simply a set of edges (as no selection of the time labels has to be made).

Under this convention, each vertex~$v$ of a temporal graph naturally
induces an ordering of its neighbors by nondecreasing values
of~$\lambda(\{v,\cdot \})$, with ties broken arbitrarily.
Throughout the paper, we refer to the $i$-th neighbor in this ordering as the $i$-th smallest neighbor of~$v$.

It is often more convenient to work with bicliques rather than cliques.
A \emph{temporal biclique} $G=(S,T,\lambda)$ consists of disjoint sets of \emph{sources}~$S$ and \emph{targets}~$T$, together with a time-labeling function~$\lambda$.
Its vertex set is $S\cup T$, while its edge set is omitted from the notation, as it is always $E(G):=\{\{s,t\}: s\in S, t\in T\}$.
We say that $G$ is  \emph{balanced of size~$n$} if $|S|=|T|=n$.
Moreover, for $S'\subseteq S$ and $T'\subseteq T$, we denote by
$G[S',T']$ the \emph{induced temporal subgraph}
$(S',T',\lambda |_{\{\{s,t\}: s\in S',t\in T'\}})$.
A temporal \emph{$\alpha$-bispanner} of a temporal biclique~$G$ is a set $E'\subseteq E(G)$
such that, for every $s\in S$ and $t\in T$, there exists a temporal path from~$s$ to~$t$ with at most~$\alpha$ hops using only edges in~$E'$; paths from targets to sources are not required.

Importantly, it was observed in~\cite{ABF+24} that the questions of finding sparse spanners for temporal cliques and temporal bicliques are essentially equivalent.
This allows us to work with bicliques rather than cliques, and we prove bipartite analogues of \cref{thm:3-spanners,thm:lifetime}.
Using the construction from~\cite{ABF+24}, it is then straightforward to derive our results for temporal cliques.
For completeness, we provide the derivation below.

\begin{observation}
\label{obs:biclique}
Assume every balanced temporal biclique of size $n$ and lifetime $L=L(n)$ admits an $\alpha$-bispanner of size $f(n, L)$. Then every temporal clique on $n$ vertices with lifetime~$L$ admits an $\alpha$-spanner of size $f(n, L)$.
\end{observation}

\begin{proof}
Let $G=(V,E,\lambda)$ be a temporal clique with $|V|=n$ and time labels in $\{1, \dots, L\}$.
We construct a temporal biclique $B=(S,T,\lambda')$ as follows (cf.~\cref{fig:bicliques}): We let $S,T$ be copies of $V$, denoted $S:=\{u_S: u\in V\}$, $T:=\{u_T: u\in V\}$, and set
\[
	\lambda'(\{u_S,v_T\}) = \begin{cases*}
		\lambda(\{u,v\}) , & if $u \neq v$,\\
		1,            & if $u=v$.
	\end{cases*}
\]
Let $E_B'$ be a temporal $\alpha$-bispanner of $B$.
Set $E_G'\subseteq E(G)$ to be the set~$E_G':=\{\{u,v\} : u\neq v, \{u_S,v_T\}\in E_B'\}$.
Then $|E_G'|\leq |E_B'|$. Moreover, for $u_S \in S, v_T \in T$ with $u \neq v$, there is a temporal $u_S$-$v_T$-path $P_B$ of at most $\alpha$ hops
in $B$ using the edges of $E'_B$. This path $P_B$ corresponds to
a temporal $u$-$v$-path $P_G$ in $G$ using the edges $E_G'$,
where every edge $\{x_S,y_T\}$ of $P_B$ such that $x=y$ is compressed
to a single vertex $x=y$.
Observe that the number of edges of $P_G$ is upper bounded by the number of edges of~$P_B$, as the consecutive copies of the same vertex of $G$ are merged.
\end{proof}

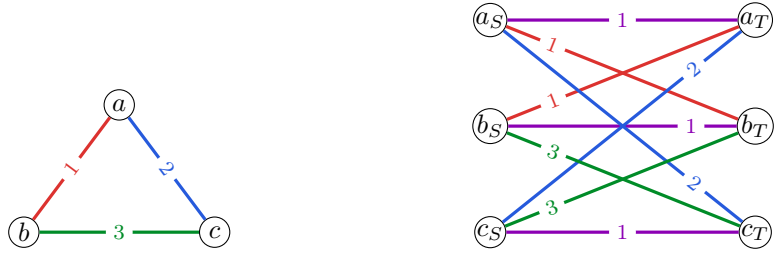
\begin{figure}
\centering
\begin{tikzpicture}[
    vertex/.style={
        circle,
        draw,
        fill=white,
        minimum size=4mm,
        inner sep=0pt,
        font=\small
    },
    edgeLabel/.style={
        circle,
        fill=white,
        inner sep=1.2pt,
        font=\scriptsize,
        sloped
    }, scale=0.7
]
\definecolor{myred}{RGB}{220,50,47}
\definecolor{myblue}{RGB}{38,90,220}
\definecolor{mygreen}{RGB}{0,140,40}
\definecolor{myviolet}{RGB}{140,0,180}

\node[vertex] (a) at (0,2.4) {$a$};
\node[vertex] (b) at (-1.8,0) {$b$};
\node[vertex] (c) at (1.8,0) {$c$};

\draw[very thick,myred]
    (a)--node[edgeLabel,text=myred,midway]{1}(b);

\draw[very thick,myblue]
    (a)--node[edgeLabel,text=myblue,midway]{2}(c);

\draw[very thick,mygreen]
    (b)--node[edgeLabel,text=mygreen,midway]{3}(c);

\node[vertex] (aL) at (7,4) {$a_S$};
\node[vertex] (bL) at (7,2) {$b_S$};
\node[vertex] (cL) at (7,0) {$c_S$};

\node[vertex] (aR) at (12,4) {$a_T$};
\node[vertex] (bR) at (12,2) {$b_T$};
\node[vertex] (cR) at (12,0) {$c_T$};

\draw[very thick,myviolet]
    (aL)--node[edgeLabel,text=myviolet,midway]{1}(aR);

\draw[very thick,myviolet]
    (bL)--node[edgeLabel,text=myviolet,pos=.80]{1}(bR);

\draw[very thick,myviolet]
    (cL)--node[edgeLabel,text=myviolet,midway]{1}(cR);

\draw[very thick,myred]
    (aL)--node[edgeLabel,text=myred,pos=.20]{1}(bR);

\draw[very thick,myred]
    (bL)--node[edgeLabel,text=myred,pos=.20]{1}(aR);

\draw[very thick,myblue]
    (aL)--node[edgeLabel,text=myblue,pos=.80]{2}(cR);

\draw[very thick,myblue]
    (cL)--node[edgeLabel,text=myblue,pos=.80]{2}(aR);

\draw[very thick,mygreen]
    (bL)--node[edgeLabel,text=mygreen,pos=.20]{3}(cR);

\draw[very thick,mygreen]
    (cL)--node[edgeLabel,text=mygreen,pos=.20]{3}(bR);

\end{tikzpicture}

\caption{A temporal clique $G$ on the left and the corresponding
biclique $B$ on the right. The temporal path $P_B= (a_S, a_T, b_S, c_T)$
 in $B$ corresponds to the path $P_G= (a,b,c)$ in $G$. }
\label{fig:bicliques}
\end{figure}
 \section{Almost linear 3-bispanners of temporal bicliques} \label{sec:3bispan}

In this section, we present an algorithm which constructs a temporal $3$-bispanner of a temporal biclique. More precisely, we prove the following result.
\begin{theorem}
\label{thm:3-bispanners}
Every balanced temporal biclique of size $n$ admits a temporal 3-bispanner of size at most~$n^{1+2/\sqrt{\ln n}}=n^{1+o(1)}$. Moreover, such a bispanner can be computed in time~$\bigO(n^{2+o(1)})$.
\end{theorem}

By \cref{obs:biclique}, this in turn implies our main result, \cref{thm:3-spanners}.
The first ingredient is the algorithm~$\Z$, which, given a
temporal biclique $G=(S,T,\lambda)$,
computes a subset~$Z$ of edges in $G$
covering a large portion of pairs in $S \times T$, 
that is, for $S'\subseteq S, T'\subseteq T$,
we say that $Z$ \emph{covers} $S'\times T'$ if,
for every $s\in S'$ and $t\in T'$, there exists a temporal
$s$-$t$-path of length at most $3$ in $G$ only using edges in $Z$.
The size of the set of covered pairs depends
on the parameter $k$ given as an input to the $\Z$ algorithm.
To be more precise, the $\Z$ algorithm receives a
balanced biclique $G$ of size $n$
and a parameter $k$, and returns a triple~$(S_1,T_2,Z)$,
where $Z$ is a set of edges covering $S_1\times T_2$.
In addition to that, the algorithm guarantees that $|S_1|=k$ and $|T_2|=n-k$.
Importantly, $Z$ forms a tree of diameter~$3$, thus~$|Z| \leq n$.

The $\Z$ algorithm is fairly simple (cf.~\cref{alg:zorro}).
Intuitively, the algorithm aims to find $s^*\in S$, and $t^*\in T$, such that for every $s\in S_1$ and $t\in T_2$, the path $(s,t^*,s^*,t)$ is temporal and then sets $Z$ to be the edges in these paths, which is the union of two stars.
For this, the algorithm first chooses $t^*$ such that it belongs to the $k$ smallest neighbors for many sources, and sets $S_1$ to be a subset of such sources.
Next, it chooses~$s^*\in S_1$ to be the largest neighbor of $t^*$ amongst the vertices in $S_1$.
Last, it sets $T_2$ to be the $n-k$ largest neighbors of $s^*$.
The procedure is formally defined in \cref{alg:zorro} and illustrated in \cref{fig:zorro}.
In the following lemma, we establish its correctness and key properties.
\vspace{2mm}

\begin{algorithm}[t]
	\caption{$\Z(G,k)$ \label{alg:zorro}}
    \KwIn{balanced temporal biclique $G=(S,T,\lambda)$ of size $n$, and $k\in\{1,\dots, n-1\}$
}
    \KwOut{$(S_1,T_2,Z)$ where $S_1\subseteq S, T_2\subseteq T, Z \subseteq E(G)$
    }
    \For{$s\in S$}{
        let $T_s$ denote the $k$ smallest neighbors of $s$\;
    }
    \label{line:x_hitting} choose $t^* \in T$ such that $|\{ s \in S \mid  t^* \in T_s\}| \ge k$\;
    choose $S_1 \subseteq \{s \in S \mid t^* \in T_s\}$ such that $|S_1| = k$\;
    let $s^* \in S_1$ be the largest neighbor of $t^*$ in $S_1$\;
    let $T_2$ be the set of $n-k$ largest neighbors of $s^*$\;
    $Z \gets \{ \{s,t^*\}: s \in S_1 \} \cup \{ \{s^*,t\} : t \in T_2 \}$ \;
    \Return $(S_1,T_2,Z)$\;
\end{algorithm}

\begin{lemma}\label{lem:zorro}
The algorithm $\Z$ is well-defined and, given a
balanced temporal biclique $G = (S,T,\lambda)$ of size $n$ and a parameter $k \in \{1,\dots, n-1\}$, it outputs a triple
$(S_1, T_2, Z)$, where $S_1 \subseteq S, T_2 \subseteq T,
Z \subseteq E(G)$,
satisfying the following properties:
\begin{itemize}
\item $|S_1|=k$ and $|T_2|=n-k$, 
\item $|Z|\leq n$,
\item $Z$ covers $S_1 \times T_2$.
\end{itemize}
Moreover, the running time of $\Z$ is $\mathcal O(n^2 \log n)$.
\end{lemma}

\begin{proof}
Consider the formal description of $\Z(G,k)$ given in \cref{alg:zorro}, and let $(T_s)_{s\in S}$, $s^*$, $t^*$, $S_1$, $T_2$, and $Z$ be as defined by the algorithm.
We first argue that~$t^*$ in line~3 is well-defined. 
For this, note that $E':=\bigcup_{s \in S} \{ \{s, t\}  : t \in T_s \}$ is a set of exactly~$kn$ edges. 
By pigeonhole principle, there is a vertex in~$T$ that is incident to at least~$k$ of these edges, and this vertex is a valid choice for~$t^*$.
It is immediate that the remainder of the algorithm is also well-defined.

By definition, we have that $|S_1|=k$ (line 4) and $|T_2|=n-k$ (line 6), which together imply that $|Z|\leq k+(n-k)=n$ (line 7).
To show that $Z$ covers $S_1 \times T_2$, fix arbitrary~$s\in S_1$ and $t\in T_2$. We show that $Z$ contains a temporal $s$-$t$-path of length at most~3.
If $s=s^*$ or~$t=t^*$, then $Z$ contains the edge $\{s,t\}$, i.e., a temporal $s$-$t$-path of length~1.
(One can even observe that the case $t=t^*$ cannot occur, as $t^*\notin T_2$.)
If $s\neq s^*$ and~$t\neq t^*$, observe that the path $(s,t^*,s^*,t)$ (cf.~\cref{fig:zorro}) has length 3, is contained in $Z$ and we claim that it is temporal.
Indeed, we have $\lambda(\{s,t^*\}) \le \lambda(\{s^*,t^*\})$ because $s^*$ was chosen to be the largest neighbor of $t^*$ amongst the vertices of $S_1$.
We have $\lambda(\{s^*,t^*\}) \le \lambda(\{s^*,t\})$ because~$t^*\in T_{s^*}$, i.e., $t^*$ is one of the $k$ smallest neighbors of $s^*$, while $t\in T_2$, i.e., $t$ is one of the $n-k$ largest neighbors of $s^*$.

It remains to prove the bound on the running time.
For this, note that lines 1--4 can be implemented as follows: 
For each $t \in T$, we initialize a counter with value 0 and an empty list.
For each $s\in S$, we sort its incident edges, which requires time $\bigO(n \log n)$, then add~1 to the counters of the $k$ smallest neighbors of $s$ and add $s$ to their lists.
Carrying this out for all $n$ vertices in $S$ requires a total of $\bigO(n^2 \log n)$ operations.
Then choose $t^*$ to be a vertex in $T$ with maximum counter, computable in time $\bigO(n)$, and choose $S_1$ to be a subset of size~$k$ of the list for $t^*$.
Then, finding $s^*$ in line 5 requires time $\bigO(n)$, finding $T_2$ requires at most $\bigO(n \log n)$ operations (for sorting), and $Z$ can be retrieved in time $\bigO(n)$.
The overall time complexity is therefore dominated by $\mathcal O(n^2 \log n)$.
\end{proof}

    \begin{figure}
        \begin{center}
        \includegraphics[scale=1.15]{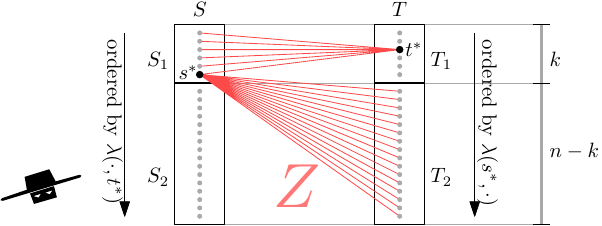}
\vspace{1mm}
\end{center}
      \caption{Illustration of the Zorro algorithm. For every $s\in S_1$ and $t\in T_2$, the path $(s,t^*,s^*,t)$ is temporal. Its Z-shape gives the algorithm its name.}\label{fig:zorro}
    \end{figure}

The rough idea for computing sparse spanners is to apply the algorithm $\Z$ and then recurse on the parts of the graph whose temporal reachabilities are not covered by its output.
This approach is formalized in \cref{alg:bispan}.
To obtain almost linear spanners, it is crucial to choose the parameter~$k$ with which~$\Z$ is invoked carefully.
For now, we allow a general choice of~$k$, viewed as a function of~$n$, and later specify the optimal choice.
Accordingly, \cref{alg:bispan} is parameterized by a function
$f\colon \mathbb{N}_{\geq 3}\to\mathbb{N}$ satisfying
$f(n)\in\{1,\dots, n-1\}$ for every $n\in \mathbb{N}_{\geq 3}$,
and we assume that $f(n)$ can be computed efficiently enough, i.e.,
within time~$\bigO(n^2\log n)$,
so that the running time of a single call to $\BS_f$ is dominated by the execution of $\Z$.
In the following lemma, we establish the correctness of the
algorithm and derive a recursive bound on the size of the resulting
spanner.

\flushbottom
\begin{samepage}
\begin{lemma}\label{thm:bispan}
Let $f\colon \mathbb{N}_{\geq 3} \to \mathbb{N}$ with $f(n)\in\{1,\dots, n-1\}$ for all $n\in \mathbb{N}_{\geq 3}$. Then, for any balanced temporal biclique $G=(S,T,\lambda)$, the algorithm $\BS_f(G)$ outputs a $3$-bispanner of $G$.
Moreover, let $$K(n):=\max\{|\BS_f(G)| : G \text{ is a balanced temporal biclique of size } n\}$$ denote the worst-case size of the computed spanner.
Then we have 
\begin{align}
	K(n) & \leq 2n + 2\cdot K(f(n)) + K(n - f(n)), &\text{for $n\geq 3$},\label{eq:K_condition_rec} \\
    K(n) &\leq n^2,                             &\text{for every $n\geq 1$}. \label{eq:K_condition_base}
\end{align}
\end{lemma}
\end{samepage}

\begin{proof}
Consider the formal description of $\BS_f$ given in \cref{alg:bispan} and let $S_1$, $T_2$, $Z$, $S_1'$, $T_2'$, and $Z'$ be as defined by the algorithm.
First, observe that all lines of the algorithm are well-defined: In particular, in line 6, all recursive calls are made on balanced bicliques.
Moreover, these all have size strictly less than $n$ so that the algorithm terminates.

    We prove by strong induction on $n$ that the output is a 3-bispanner of $G$.
	First, if $n\leq 2$, the algorithm outputs the entire edge set, which is a 1-bispanner of $G$.
    For the inductive step, assume $n\geq 3$ and let $s\in S$ and $t\in T$.
By \cref{lem:zorro} and the induction hypothesis, we obtain that there exists an $s$-$t$-path of length at most 3 in
\begin{itemize}
\item $Z\cup \BS_f(G[S_1,T\setminus T_2])$ if $s\in S_1$,
\item $Z'\cup \BS_f(G[S\setminus S_1',T'_2])$ if $t\in T_2'$,
\item  $\BS_f(G[S\setminus S_1, T\setminus T_2'])$ if $s\notin S_1$ and $t\notin T_2'$.
\end{itemize}
This covers all cases, so that the output is indeed a valid 3-bispanner of $G$.
For the bounds on~$K$, \eqref{eq:K_condition_base} follows from the fact that the output is a subset of the $n^2$ edges of $G$, and \eqref{eq:K_condition_rec} follows from $|Z\cup Z'|\leq 2n$, $|S_1|=|S\setminus S_1'|=f(n)$, and $|S\setminus S_1|=n-f(n)$.
\end{proof}

\begin{algorithm}[t]
	\caption{$\BS_f(G)$ \label{alg:bispan}} 
	\KwIn{temporal biclique $G = (S,T,\lambda)$ with $|S|=|T|=:n$}
	\KwOut{set of edges $E' \subseteq E(G)$ that is a $3$-bispanner of $G$}
	\If{$n \leq 2$}{
			\Return $E(G)$\;
		}
		$k \gets f(n)$\;
		$(S_1,T_2,Z) \gets \Z(G,k)$\;
		$(S'_1,T'_2,Z') \gets \Z(G,n-k)$\;
		\Return $Z\cup Z'\cup \BS_f(G[S_1,T\setminus T_2]) \cup \BS_f(G[S\setminus S_1',T'_2]) \cup \BS_f(G[S \setminus S_1,T\setminus T'_2])$\;
\end{algorithm}

It remains to choose the function $f$. For this, we set
\begin{equation}
\label{eq:def_f}
    f(n):=\left\lceil n e^{-\sqrt{\ln n}}\right\rceil .
\end{equation}
For every $n\geq 3$, we have $2\leq f(n)\leq n-1$, so this is a valid
choice of $f$ in \cref{thm:bispan}.
We next carefully analyze the size of the resulting spanners using the two bounds established in \cref{thm:bispan}.
We prove something slightly stronger here that will later also help us to bound the running time of the algorithm.

\begin{lemma}\label{lem:rec}
Let $H\colon \mathbb{N}\to\mathbb{N}$ be a function satisfying
\begin{align}
	H(n) & \leq 2n + 2\cdot H(f(n)) + H(n - f(n)), &\text{for $n\geq 55$},\label{eq:K_condition_rec_lem} \\
    H(n) & \leq C n^2,                             &\text{for $n\leq 54$}. \label{eq:K_condition_base_lem}
\end{align}
for some constant $C\geq 1$, where $f$ is defined as in \eqref{eq:def_f}.
Then $H(n)\leq Cn e^{2\sqrt{\ln n}}$ for every~$n\in\mathbb{N}$.
\end{lemma}

\begin{proof}
We prove the statement by strong induction on $n$.
First, let $n\leq \lfloor e^4 \rfloor = 54$, in particular $\sqrt{\ln n}\leq 2$. 
Then
\[
    H(n)\overset{\eqref{eq:K_condition_base_lem}}{\leq} C n^2
    =C n e^{\ln n}
	=C n e^{\sqrt{\ln n}\cdot \sqrt{\ln n}}
	\leq C n e^{2\sqrt{\ln n}},
\]
so the claim holds for $n\leq 54$.

Now, let $n\geq 55$ and suppose that the claim holds for every positive
integer smaller than~$n$. Put
\[
    L:=\sqrt{\ln n},\qquad \text{and} \qquad
    k:=n e^{-L},
\]
so that $f(n)=\lceil k \rceil$.
Recall that $f(n)<n$ and $n-f(n)<n$, so that we obtain
\begin{align}
H(n)&\overset{\eqref{eq:K_condition_rec_lem}}{\leq} 2n+2H(\lceil k \rceil )+H(n-\lceil k \rceil) \nonumber \\
\overset{\text{ind.~hyp.}}&{\leq} 2n + 2C \lceil k \rceil e^{2\sqrt{\ln \lceil k \rceil}} +  C(n-\lceil k \rceil) e^{2\sqrt{\ln (n-\lceil k \rceil)}}\nonumber\\
&\leq 2n + 2 C(k+1) e^{2\sqrt{\ln(k+1)}} + C(n-k)e^{2\sqrt{\ln(n-k)}},\label{eq:bound_on_K}
\end{align}
where in the last inequality, we have used that the function $x \mapsto xe^{2\sqrt{\ln x}}$ is increasing for~$x\geq 1$.
Next, we derive bounds for the last two summands of the expression \eqref{eq:bound_on_K}.

Since $n\geq 55$, it can be easily checked that $k=ne^{-\sqrt{\ln n}}\geq 4$ (in fact $n\geq 24$ suffices).
To bound the second summand of \eqref{eq:bound_on_K}, observe first that
\begin{align}
\ln(k+1)
&=\ln\left( k \left(1+\frac{1}{k}\right) \right)
= \ln n-\sqrt{\ln n} +\ln\left(1+\frac1k\right) 
    =L^2-L+\ln\left(1+\frac1k\right)\nonumber\\
    &\overset{(*)}{\leq} L^2-L+\frac1k
\hspace{2mm}\overset{k\geq 4}{\leq} \hspace{2mm}
L^2 - L + \frac{1}{4}
= \left(L-\frac12\right)^2,\label{eq:ln_kplusone}
\end{align}
where $(*)$ follows from the fact that $\ln(1+x)\leq x$ for every $x\geq 0$.
Consequently, we obtain the following bound
\begin{equation}
\label{eq:first_summand}
    (k+1) e^{2\sqrt{\ln(k+1)}}
\overset{\eqref{eq:ln_kplusone}}{\leq} (n e^{-L}+1)e^{2L-1}
    =n e^{L-1}+e^{2L-1}.
\end{equation}

We proceed with bounding the last summand of \eqref{eq:bound_on_K}.
Using the standard fact that $\ln(1-x)\leq -x$ for $x\in [0,1)$, we obtain
\begin{align}
\ln(n-k) 
&= \ln(n (1-e^{-L}))
=\ln n+\ln(1-e^{-L})\le  L^2 - e^{-L}
=\left(L-\frac{e^{-L}}{2L}\right)^2 - \left(\frac{e^{-L}}{2L}\right)^2\nonumber\\
&\leq \left(L-\frac{e^{-L}}{2L}\right)^2.\label{eq:ln_nminusk}
\end{align}
Therefore,
\begin{align}
(n-k)e^{2\sqrt{\ln(n-k)}}
&=n(1-e^{-L})e^{2\sqrt{\ln(n-k)}} \nonumber\\
&\overset{\eqref{eq:ln_nminusk}}{\leq} n(1-e^{-L})e^{2L-\frac{1}{Le^L}}
=ne^{2L}(1-e^{-L})e^{-\frac{1}{Le^L}}\nonumber\\
&\overset{(*)}{\leq} ne^{2L}\frac{1-e^{-L}}{1+\frac{1}{Le^L}}
= ne^{2L}\left( \frac{L-Le^{-L}}{L+e^{-L}}\right)
=ne^{2L}\left( 1- \frac{(L+1)e^{-L}}{L+e^{-L}}\right)\nonumber\\
&=ne^{2L} - ne^{L}\frac{L+1}{L+e^{-L}}
\overset{L\geq 2}{\leq} ne^{2L} - ne^{L}\frac{L+1}{L+1/4}\nonumber\\
&= ne^{2L} - ne^{L}\left(1+\frac{3/4}{L+1/4}\right)
,
\label{eq:second_summand}
\end{align}
where we used for $(*)$ that $e^{-x}\leq 1/(1+x)$ for every $x\geq0$, and in the next line that~$L=\sqrt{\ln n}\geq 2$ as~$n\geq 55$ (and $\sqrt{\ln 55}\thickapprox 2.0018$).

Summarizing our analysis so far, we have
\begin{align*}
H(n)
&\overset{\eqref{eq:bound_on_K}}{\leq }
2n + 2 C(k+1) e^{2\sqrt{\ln(k+1)}} + C(n-k)e^{2\sqrt{\ln(n-k)}}\\
\overset{\eqref{eq:first_summand}\eqref{eq:second_summand}}&{\leq}
2n + C \left( 2 n e^{L-1}+2 e^{2L-1} +ne^{2L} - ne^{L}\left(1+\frac{3/4}{L+1/4}\right)\right)\\
\overset{C\geq 1}&{\leq}Cne^{2L} +Cn\left(2+e^L\cdot \frac{2}{e}+2\frac{e^{2L-1}}{e^{L^2}}-e^L\left(1+\frac{3/4}{L+1/4}\right)\right)\\
&=Cne^{2L} +Cn\left(2+2e^{-(L-1)^2}-e^L\left(1+\frac{3/4}{L+1/4}-\frac{2}{e}\right)\right)\\
&\leq Cne^{2L},
\end{align*}
where the last inequality follows from $2+2e^{-(L-1)^2}-e^L\left(1+\frac{3/4}{L+1/4}-\frac{2}{e}\right)\leq 0$ for $L\geq 2$.

\end{proof}

Having established the size of the resulting spanner, it remains only to analyze the running time of~$\BS_f$ for the choice of~$f$ given in \eqref{eq:def_f}.

\begin{lemma}
\label{lem:running_time}
The running time of $\BS_f$ is $\bigO(n^{2+o(1)})$.
\end{lemma}

\begin{proof}
For an execution of $\BS_f$ on a biclique $G$, consider all the issued
recursive calls of the algorithm (which can be represented as the recursion tree).
For each call $q$ to the algorithm (i.e., a node of the recursion tree),
let $n_q$ denote the size of the balanced biclique on which
the recursion is called. Then let $A(G):=\sum_q n_q$ denote the sum
of the instance sizes, and for $n\in\mathbb{N}$, let~$A(n):=\max\{A(G): G \text{ is a balanced biclique of size } n\}$ denote the worst-case sum.

We trivially have $A(n)=\bigO (1)$ for $n\leq 55$ and observe that $A$ fulfills the recurrence~$A(n)\leq n+  2 A(f(n)) + A(n-f(n))$.
In particular, the conditions of \cref{lem:rec} are fulfilled (choosing some large enough constant $C\geq 1$), so that we have $A(n)=\bigO(n^{1+2/\sqrt{\ln n}})$.

Next, observe that the running time of a single call to $\BS_f$
on a biclique of size~$n_q$ is $\bigO(n_q^2 \log (n_q+1))$ by \cref{lem:zorro}.
Since the function $h \colon \mathbb{R}_{\geq 0} \to \mathbb{R}, x \mapsto x^2 \log (x+1)$ is superadditive (because it is convex and $h(0)=0$), we obtain that the total running time is bounded by
\begin{equation*}
\sum_q \bigO(n_q^2 \log (n_q+1))
\leq \bigO(A^2(n)\log(A(n)+1))
=
\bigO(n^{2+4/\sqrt{\ln n}}\log n)=\bigO(n^{2+o(1)}).\qedhere
\end{equation*}
\end{proof}

To summarize, \cref{thm:bispan} establishes that the algorithm $\BS_f$ outputs a valid 3-bispanner of every balanced temporal biclique,  \cref{lem:rec} (applied with $C=1$) establishes that its size is at most $ne^{2\sqrt{\ln n}}=n\cdot n^{2\sqrt{\ln n}/\ln n}=n^{1+2/\sqrt{\ln n}}$, and \cref{lem:running_time} shows that the algorithm runs in time~$\bigO(n^{2+o(1)})$.
Together, these results prove our main result for temporal bicliques, \cref{thm:3-bispanners}.
By \cref{obs:biclique}, this result carries over to temporal cliques and thus completes the proof of \cref{thm:3-spanners}.

 \section{Temporal cliques of bounded lifetime}

In this section, we turn to considering temporal bicliques of bounded lifetime $L$, meaning that edge labels belong to $\{1,\dots, L\}$, and prove (a bipartite analogue of) \cref{thm:lifetime}.
We assume throughout this section that $L\geq 2$, because otherwise, the biclique trivially contains a 3-bispanner of size $2n-1$ (the union of a star centered at a source and a star centered at a target).

The following lemma is similar in spirit to the algorithm~$\Z$.
However, while the sizes of the sets $S_1$ and $T_2$ found by~$\Z$ sum to~$n$, the bounded-lifetime assumption allows us to find sets whose total size is a constant factor larger than~$n$.

\begin{lemma}\label{lem:large-lifetime-rectangle}
Let $G=(S,T,\lambda)$ be a balanced temporal biclique of size $n$
and lifetime at most~$L$.
Then there are sets $S'\subseteq S$, $T'\subseteq T$ and
$Z\subseteq E(G)$, computable in time $\bigO(n^3)$, such that
\begin{itemize}
    \item $Z$ covers $S' \times T'$,
    \item $|Z|\leq |S'|+|T'|$, and
    \item $|S'|+|T'|\geq n+n/L$.
\end{itemize}
\end{lemma}

\begin{proof}
For $s\in S$, $t\in T$, let
\[
S_{st}:=
\{s'\in S:\lambda(\{s',t\})\leq\lambda(\{s,t\})\}
\qquad \text{and} \qquad
T_{st}:=\{t'\in T:\lambda(\{s,t'\})\geq\lambda(\{s,t\})\}.
\]
Observe that the set $Z_{st}:=\{\{s',t\} : s'\in S_{st}\} \cup \{\{s,t'\} : t'\in T_{st}\}$ has size at most~$|S_{st}|+|T_{st}|$ and covers $S_{st} \times T_{st}$: Indeed, for $s'\in S_{st}\setminus \{s\}$ and $t'\in T_{st}\setminus \{t\}$, the path $(s',t,s,t')$ is temporal and has length 3.
Therefore, it suffices to show that there exists $s,t$ such that~$|S_{st}|+|T_{st}|\geq n +n/L$.

For this, fix some $s\in S$ and our first goal is to derive a lower bound on $\sum_{t\in T}|T_{st}|$.
Take some $\{t_1,t_2\}\in \binom{T}{2}$ such that $\lambda(\{s,t_1\})\leq \lambda(\{s,t_2\})$. 
This pair contributes at least 1 to the sum as $t_2\in T_{st_1}$.
If  $\lambda(\{s,t_1\})=\lambda(\{s,t_2\})$, the pair even contributes 2 to the sum as we also have  $t_1\in T_{st_2}$.
Moreover, since $t \in T_{st}$, each $t\in T$ also contributes 1 to the sum.
Letting~$m_i$ denote the number of edges incident to $s$ of label $i$ ($i\in\{1,\dots, L\}$), we therefore obtain
\begin{align*}
\sum_{t\in T}|T_{st}|
&\geq 
n + \binom{n}{2}+\sum_{i=1}^L \binom{m_i}{2}
= n + \frac{n^2-n}{2}+ \sum_{i=1}^L \frac{m_i^2-m_i}{2}\\
\overset{\sum_i m_i=n}&{=} \hspace{2mm}
\frac{1}{2}\cdot \left(n^2 + \sum_{i=1}^{L}m_i^2\right)
\overset{\text{Cauchy-Schwarz}}{\geq}
 \frac{1}{2}\cdot \left(n^2 + \frac{n^2}{L}\right)
=\frac{n^2}{2}\left(1+\frac{1}{L}\right).
\end{align*}
By analogous argumentation, we obtain the same bound for $\sum_{s\in S} |S_{st}|$ for any fixed $t$. We obtain
\begin{align*}
\sum_{s\in S, t\in T} (|S_{st}|+|T_{st}|) \geq 2n \cdot \frac{n^2}{2}\left(1+\frac{1}{L}\right) 
= n^3\left(1+\frac{1}{L}\right).
\end{align*} 
Therefore, there exists $s\in S$ and $t\in T$ such that $|S_{st}|+|T_{st}|\geq n(1+1/L)$.
Last, observe that, for fixed $s$ and $t$, the sets $S_{st}$ and $T_{st}$ can be computed in time~$\bigO(n)$.
Therefore, a pair~$(s,t)$ maximizing $|S_{st}|+|T_{st}|$ can be found in time~$\bigO(n^3)$, after which $S', T'$, and $Z$ can be retrieved in time~$\bigO(n)$.
\end{proof}

A key difficulty when applying recursion is that we will have to handle unbalanced bicliques.
The next lemma handles the case where the difference in size is not too large and allows us to remove vertices from the larger bipartition side, at the cost of at most linearly many edges, resulting in a balanced biclique.

\begin{lemma}\label{lem:one-balancing-step}
Let $r\in\mathbb{N}$ and $G=(S,T,\lambda)$ be a temporal biclique of lifetime at most $L$
with
\[
|S|=n,\qquad |T|=n+r,
\qquad\text{and}\qquad
r<n/(L-1).
\]
Then there are a set $T_0\subseteq T$ of size $r$ and a set
$F\subseteq E(G)$ of at most $2n+r$ edges,
computable in time $\bigO(n^2)$, 
such that~$F$ covers $S\times T_0$.
The symmetric statement holds when~$|S|=n+r$ and~$|T|=n$.
\end{lemma}

\begin{proof}
Choose $s^*\in S$ maximizing
$
\sum_{t\in T}\lambda(\{s^*,t\})
$, and let $T_0$ consist of the $r$ largest neighbors of $s^*$.

We claim that, for every $s\in S$, there is some $t_s\in T\setminus T_0$ such that
$\lambda(\{s,t_s\})\leq\lambda(\{s^*,t_s\}).$
Indeed, otherwise
\[
\begin{aligned}
\sum_{t\in T}
\bigl(\lambda(\{s,t\})-\lambda(\{s^*,t\})\bigr)
&\geq \sum_{t\in T\setminus T_0}1 + \sum_{t\in T_0} \bigl(\lambda(\{s,t\})-\lambda(\{s^*,t\})\bigr)\\
&\geq |T\setminus T_0| + |T_0| (1-L)
=n-(L-1)r>0,
\end{aligned}
\]
contradicting the choice of $s^*$.

Let $F$ be the set of edges $\{s,t_s\}$ and $\{s^*,t_s\}$ for every $s\in S$,
and all edges between~$s^*$ and~$T_0$. Then $|F|\leq 2n+r$ and it only remains to argue that $F$ covers~$S \times T_0$.
For $s\in S$ and~$t\in T_0$, if $F$ does not already contain the direct edge $\{s,t\}$, observe that the path $(s,t_s,s^*,t)$ is temporal.
This is because $\lambda(\{s,t_s\})\leq\lambda(\{s^*,t_s\})$ by definition of~$t_s$ and~$\lambda(\{s^*,t_s\}) \leq \lambda(\{s^*,t\})$ as $t\in T_0$ and~$t_s\notin T_0$, that is, $t$ belongs to the $r$ largest neighbors of $s^*$, while $t_s$ does not.

Observe that $s^*$ can be computed in time $\bigO(n^2)$ and $T_0$ in time $\bigO(n\log n)$ (for sorting). For each $s\in S$, the target $t_s$ can be found in time~$\bigO(n)$, resulting in a total running time of~$\bigO(n^2)$ for computing $(t_s)_{s\in S}$. The set~$F$ can then be retrieved in time~$\bigO(n)$. Therefore, the total running time is~$\bigO(n^2)$.

The symmetric statement follows by interchanging $S$ and $T$,
replacing each label $i$ by~$L+1-i$, and reversing the resulting
paths.
\end{proof}

The next result builds on the previous lemma and handles the case of a larger difference between the sizes of $S$ and $T$.

\begin{lemma}\label{lem:balancing}
Let $r\in\mathbb{N}$ and $G=(S,T,\lambda)$ be a temporal biclique of lifetime at most $L$
with
\[
|S|=n,\qquad |T|=n+r.
\]
Then there are a set $T_0\subseteq T$ of size $r$ and a set
$F\subseteq E(G)$, computable in time $\bigO(rn^2)$, such that~$F$ covers $S \times T_0$ and
\[
|F|\leq 3n+(4L+1)r.
\]
The symmetric statement holds when $|S|=n+r$ and $|T|=n$.
\end{lemma}

\begin{proof}
If $r<n/(L-1)$, apply \cref{lem:one-balancing-step} directly.
So assume from now on that $r\geq n/(L-1)$.
If $n<L$, choose any $r$ vertices of $T$ for
$T_0$ and let $F$ consist of all edges between $S$ and $T_0$. 
This trivially covers $S \times T_0$ and $|F|\leq nr\leq Lr$, so the statement holds.

It remains to consider $n\geq L$. Set
\[
q:=\left\lfloor\frac{n-1}{L-1}\right\rfloor\geq1.
\]
Choose an arbitrary $r$-element set $T'\subseteq T$ and partition it into at most $\lceil r/q\rceil$ batches, each of size at most $q$, denoted $B_1,\dots, B_{\lceil r/q\rceil}$.
Then, apply the following procedure: begin with~$T^*=T \setminus T'$, i.e., $|T^*|=n$. In the $i$-th step, first add $B_i$ to $T^*$, obtaining a target set of size $n+q$, apply \cref{lem:one-balancing-step}, add the computed edge set to $F$ and remove the computed set~$T_0$ from $T^*$ to again obtain a target set of size $n$.
Note that the application of \cref{lem:one-balancing-step} is valid as $q\leq (n-1)/(L-1)<n/(L-1)$.
We finally set $T_0$ to be the set of removed targets.
By \cref{lem:one-balancing-step}, the resulting set $F$ covers $S \times T_0$ and we have
\begin{equation*}
|F|\leq \left\lceil \frac{r}{q} \right\rceil \cdot (2n+q)
\leq \left(\frac{r}{q}+1\right) (2n+q)
= 2n+q+r+\frac{2nr}{q}
\overset{(*)}{\leq}
3n + (4L+1)r,
\end{equation*}
where the last inequality $(*)$ follows from $q\leq n$ and we claim that $n/q\leq 2L$.
Indeed, we have
$
2q\geq q+1>\frac{n-1}{L-1}.
$
Since all quantities in this inequality are integers, this implies $n\leq 2q(L-1)\leq 2qL$, which completes the size bound of $F$.

For the running time claim, observe that each of the $\lceil r/q \rceil\leq r$ applications of  \cref{lem:one-balancing-step} requires time $\bigO(n^2)$, resulting in a total running time of $\bigO(rn^2)$.

\end{proof}

Now, we have all the prerequisites at hand to recursively construct sparse 3-bispanners of bicliques of bounded lifetime.

\begin{theorem}\label{thm:bounded-lifetime-biclique}
Every balanced temporal biclique of size $n$ and lifetime at most $L$
admits a temporal $3$-bispanner of size at most $12Ln$, computable in time $\bigO(Ln^3)$.
\end{theorem}

\begin{proof}
We proceed by induction on $n$. The claim is clear for $n=1$.

For $n\geq 2$, apply \cref{lem:large-lifetime-rectangle} and let $S'$, $T'$, and $Z$ be the obtained sets. Set $r:=|S'|+|T'|-n\geq n/L$. Then $Z$ covers $S' \times T'$ and
\begin{equation}\label{eq:size_of_Z}
|Z|\leq |S'|+|T'|=n+r.
\end{equation}
Since
$
    |S'|=r+n-|T'|=|T\setminus T'|+r,
$
\cref{lem:balancing}, applied to~$G[S',T\setminus T']$, gives a set~$S_0\subseteq S'$ of size $r$
and a set $F_S$ that covers $S_0 \times T\setminus T'$, where
\begin{equation}\label{eq:size_of_FS}
    |F_S|
    \leq 3(n-|T'|)+(4L+1)r.
\end{equation}
Similarly, since
$
    |T'|=|S\setminus S'|+r,
$
\cref{lem:balancing} applied to $G[S\setminus S',T']$, gives a set
$T_0\subseteq T'$ of size $r$ and a set $F_T$ that covers $S\setminus S' \times T_0$, where
\begin{equation}\label{eq:size_of_FT}
    |F_T|
    \leq 3(n-|S'|)+(4L+1)r.
\end{equation}

We claim that every pair $s,t$ with $s\in S_0$ or $t\in T_0$ is covered by $Z\cup F_S\cup F_T$.
Indeed, if~$s\in S_0$, then the pair is covered by $Z$ when its target belongs to $T'$, and by $F_S$ otherwise.
Similarly, if $t\in T_0$, then the pair is covered by $Z$ when its source belongs to $S'$, and by $F_T$ otherwise.

It therefore remains only to recursively span
$
    G[S\setminus S_0,T\setminus T_0],
$
which is balanced of size~$n-r\leq n-n/L$. 
Before the recursive call, the number of edges added to the spanner is at most
\begin{align*}
|Z|+|F_S|+|F_T|
\overset{\eqref{eq:size_of_Z}\eqref{eq:size_of_FS}\eqref{eq:size_of_FT}}&{\leq} n+r
   +3(n-|T'|)+(4L+1)r
   +3(n-|S'|)+(4L+1)r\\
&=7n+(8L+3)r-3(|T'|+|S'|)
=4n+8Lr
\leq12Lr,
\end{align*}
where the last equality uses $|S'|+|T'|=n+r$, and the last inequality uses
$r\geq n/L$.

By the induction hypothesis, the total number of edges is therefore
at most
\[
    12Lr+12L(n-r)=12Ln.
\]
By \cref{lem:large-lifetime-rectangle,lem:balancing} and since $r\leq n$, the sets $Z,F_S,F_T$ can be computed in time~$\bigO(n^3)$.
Hence, the running time~$R(n)$ of the procedure satisfies the recurrence $R(n)\leq \bigO(n^3)+R((1-1/L)n)$.
Summing the resulting geometric series yields $R(n)=\bigO(Ln^3)$.
\end{proof}

Together with \cref{obs:biclique}, this completes the proof of \cref{thm:lifetime}.
 
\bibliographystyle{alphaurl}
\bibliography{bib}

\end{document}